\PassOptionsToPackage{table}{xcolor}
\PassOptionsToPackage{protrusion=true,expansion=false}{microtype}
\documentclass[]{aidata}

\usepackage[toc,page,header]{appendix}
\usepackage{minitoc}
\usepackage{latexsym}
\usepackage[utf8]{inputenc}
\usepackage{inconsolata}
\usepackage{booktabs}
\usepackage{makecell}
\usepackage{array}
\usepackage{tabularx}
\usepackage{ragged2e}
\usepackage{amsmath,amssymb,mathtools,amsthm}
\usepackage{algorithm}
\usepackage{algpseudocode}
\usepackage[most]{tcolorbox}
\usepackage{fvextra}
\usepackage{tikz}
\usetikzlibrary{fit,calc}
\usepackage{enumitem}

\newcolumntype{C}{>{\Centering\arraybackslash}X}

\theoremstyle{plain}
\newtheorem{theorem}{Theorem}[section]

\theoremstyle{definition}

\theoremstyle{remark}

\definecolor{OverallCyan}{HTML}{F7FEFE}
\definecolor{OverallGray}{HTML}{F5F5F5}
\definecolor{BlockLavender}{HTML}{EFEAFB}
\definecolor{OursPink}{HTML}{F7DDE8}
\definecolor{GainGreen}{HTML}{0B7D2B}
\definecolor{RuleGray}{HTML}{BDBDBD}
\colorlet{algorange}{orange!30}
\colorlet{algpink}{OursPink}
\colorlet{algpurple}{BlockLavender}

\tcbset{
  appendixbox/.style={
    enhanced,
    breakable,
    colback=gray!2,
    colframe=gray!35,
    boxrule=0.35pt,
    arc=1pt,
    left=4pt,
    right=4pt,
    top=3pt,
    bottom=3pt,
    fonttitle=\bfseries
  }
}

\newtcbox{\hlprimarytab}{on line, rounded corners, box align=base,
  colback=green!10, colframe=white, size=fbox, arc=3pt,
  before upper=\strut, top=-2pt, bottom=-4pt, left=-2pt, right=-2pt, boxrule=0pt}

\newtcbox{\hlsecondarytab}{on line, rounded corners, box align=base,
  colback=red!10, colframe=white, size=fbox, arc=3pt,
  before upper=\strut, top=-2pt, bottom=-4pt, left=-2pt, right=-2pt, boxrule=0pt}

\newcommand{\dashifted}{\raisebox{0.5\depth}{\tiny$\downarrow$}}
\newcommand{\uashifted}{\raisebox{0.5\depth}{\tiny$\uparrow$}}
\newcommand{\dar}[1]{{\raisebox{0.6ex}{\tiny\hlsecondarytab{\dashifted\,#1}}}}
\newcommand{\uar}[1]{{\raisebox{0.6ex}{\tiny\hlprimarytab{\uashifted\,#1}}}}
\newcommand{\nodiff}{{\raisebox{0.6ex}{\tiny\hlprimarytab{\textemdash}}}}
\newcommand*{\tikzmk}[1]{\tikz[remember picture,overlay,] \node (#1) {};\ignorespaces}
\newcommand{\boxit}[1]{\tikz[remember picture,overlay]{\node[xshift=-7.5em,yshift=-1em,fill=#1,opacity=.25,fit={(A)($(B)+(1.1\linewidth,1\baselineskip)$)}] {};}\ignorespaces}
\newcommand{\boxittwo}[1]{\tikz[remember picture,overlay]{\node[xshift=-4.15em,yshift=-1em,fill=#1,opacity=.25,fit={(A)($(B)+(1.025\linewidth,1\baselineskip)$)}] {};}\ignorespaces}
\newcommand{\boxithree}[1]{\tikz[remember picture,overlay]{\node[xshift=-4.15em,yshift=-1em,fill=#1,opacity=.25,fit={(A)($(B)+(0.475\linewidth,1\baselineskip)$)}] {};}\ignorespaces}
\newcommand{\vnum}[1]{\makebox[3.6em][l]{#1}}
\newcommand{\varr}[1]{\makebox[2.2em][l]{#1}}
\newcommand{\ogray}[1]{\cellcolor{OverallGray}#1}

\graphicspath{{figures/}}

\title{Socratic-SWE: Self-Evolving Coding Agents via Trace-Derived Agent Skills}

\author[2,*]{Chuan Xiao}
\author[1,2,*]{Zhengbo Jiao}
\author[2]{Shaobo Wang}
\author[1]{Wei Wang}
\author[1]{Bing Zhao}
\author[1,\dagger]{HU WEI}
\author[2,\dagger]{Linfeng Zhang}
\author[1]{Lin Qu}

\affiliation[1]{AI Data, Alibaba Group}
\affiliation[2]{Shanghai Jiao Tong University}

\contribution[*]{Equal contribution}
\contribution[\dagger]{Corresponding authors}

\abstract{LLM-driven software engineering agents have become a central testbed for
real-world language-model capability, yet their training remains limited
by the availability of high-quality SWE tasks. Existing synthetic data
methods typically create tasks through fixed mutation or bug-injection
procedures, making the resulting distributions largely independent of
the agent's own weaknesses and training progress. We introduce
\textsc{Socratic-SWE}, a closed-loop self-evolution framework that
reuses the agent's historical solving traces as a source of training
signal. Rather than treating traces only as evidence for reward
computation, \textsc{Socratic-SWE} distills them into structured agent
skills that summarize recurring failures and effective repair patterns.
These skills then guide the generation of targeted repair tasks in real
repositories. Candidate tasks are checked through execution-based
validation and scored with a solver-gradient alignment reward, so that
the retained tasks are both verifiable and useful for improving the
Solver. The updated Solver produces new traces, enabling the task
curriculum to adapt over successive rounds. Across SWE-bench Verified,
SWE-bench Lite, SWE-bench Pro, and Terminal-Bench 2.0,
\textsc{Socratic-SWE} consistently improves over self-evolving baselines
under the same compute budget, reaching \textbf{50.40\%} on SWE-bench
Verified after three iterations. These results suggest that solving
traces can serve as a scalable substrate for self-evolving SWE agents.}

\date{\today}

\begin{document}

\maketitle

\begin{figure}[tb!]
  \centering
  \vspace{-0.6em}
  \includegraphics[width=0.98\linewidth]{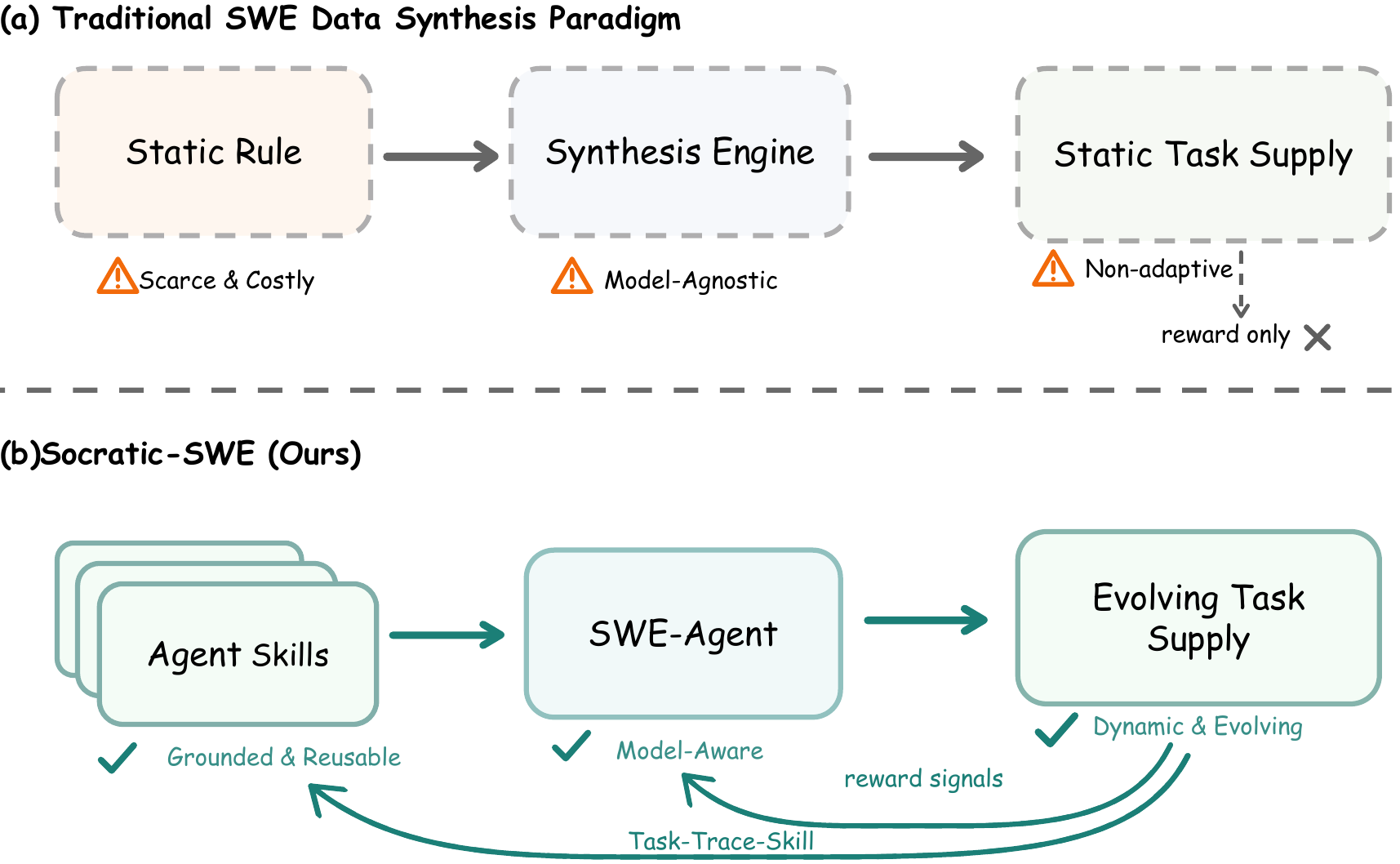}
  \vspace{-0.7em}
  \caption{
Comparison of SWE data synthesis paradigms. (a) Traditional: an open-loop, single-pass pipeline where static rules synthesize a fixed task supply. Evaluation traces are used only post hoc for rewards, not fed back into task construction, leaving the supply scarce, model-agnostic, and non-adaptive. (b) Socratic-SWE: a closed-loop, self-evolving process that distills grounded skills from solving traces to guide repository-grounded task construction, while feeding rewards into subsequent skill and task updates.
}
  \vspace{-0.7em}
  \label{fig:motivation}
\end{figure}

\section{Introduction}

Software engineering is one of the most impactful application domains 
for large language models (LLMs), and the capability of coding agents 
has emerged as a critical measure of real-world intelligence. Unlike 
mathematical reasoning or short-form code generation, SWE tasks require 
agents to complete a full loop of bug localization, repair, and 
verification through long-horizon interaction with real code 
repositories, making SWE a natural setting for training capable 
language agents in realistic environments. Reinforcement learning (RL) 
has shown strong promise in this setting: SWE-RL~\citep{wei2025swerl} 
and SWE-Gym~\citep{pan2025swegym} demonstrate substantial gains by 
embedding agents in executable repository environments. However, RL 
training fundamentally depends on large quantities of high-quality 
tasks, and such data remains scarce in the SWE domain. Existing 
synthetic pipelines attempt to alleviate this bottleneck through 
AST-level mutation, LM-guided rewriting, or learned bug 
injection~\citep{yang2025swesmith,allamanis2021bugrepair}, but 
they operate independently of the agent's own training experience, 
resulting in largely static task distributions that may be poorly 
aligned with the model's actual capability gaps.

At the same time, each round of RL training produces a valuable 
byproduct: solving traces. These traces record the agent's behavior 
throughout repository interaction, including code search, file editing, 
command execution, and test runs. They reveal where the agent repeatedly 
fails, which repair strategies tend to induce regressions, and which 
repository patterns lead to ineffective exploration. Yet existing 
methods use traces primarily for reward extraction or credit assignment. 
GRPO~\citep{shao2024deepseekmath} reduces a 
trajectory to a scalar reward, while process-level approaches such as 
GiGPO~\citep{feng2026gigpo}, iStar~\citep{liu2025agentic}, and process 
reward models~\citep{lightman2024lets} assign finer-grained credit 
within trajectories. Regardless of granularity, traces are discarded once supervision has been computed. As the model improves, the 
fraction of tasks in a fixed distribution that still provides useful 
training signal becomes increasingly sparse, and learning eventually 
stagnates.

We argue that these discarded traces contain the signals needed 
to evolve the curriculum. Because the traces are generated by the agent 
itself, they provide a direct view of the model's current capability 
boundary. This suggests a self-evolving loop: the agent distills 
capabilities from historical solving traces and uses them to construct 
the next round of training tasks, without requiring external annotation. 
Related ideas have shown promise in adjacent domains. R-Zero~\citep{huang2026rzero} and 
Socratic-Zero~\citep{wang2025socraticzero} adapt task proposers to the 
solver's frontier, Absolute Zero~\citep{zhao2025absolutezero} achieves 
zero-data self-improvement through execution-based verification, and 
SkillRL~\citep{xia2026skillrl} and SKILL0~\citep{lu2026skill0} distill 
reusable skills from interaction traces. However, these settings involve 
simple traces, such as symbolic reasoning chains, short programs, or 
finite-step games, where pass/fail feedback is often sufficient to drive 
adaptation. SWE traces are richer: a single trajectory may span dozens 
of tool invocations across search, editing, execution, and testing, 
exposing diverse and diagnosable failure modes. Combined with 
deterministic execution-based verification from repository test suites, 
this makes SWE a particularly suitable domain for trace-based self-evolution.

To address this challenge, we propose \textsc{Socratic-SWE}, a 
\textbf{self-evolution framework for software engineering}. 
Figure~\ref{fig:motivation} summarizes the closed-loop design of
\textsc{Socratic-SWE}.
\textsc{Socratic-SWE} operates in a three-stage loop. First, it distills 
recurring failure modes and successful repair behaviors from historical 
solving traces into an \emph{Agent Skill Registry}, a structured 
representation of the model's capability boundary. Second, a Generator 
uses these skills as constraints to construct targeted repair tasks in 
real code repositories. Candidate tasks are filtered through a staged 
execution-based validation pipeline to ensure reproducibility and 
non-triviality, and are further scored by a 
\emph{solver-gradient alignment reward} that favors tasks whose induced 
Solver updates align with trusted validation gradients. Third, the 
Solver trains on the accepted tasks and produces new traces, which 
support the next round of skill distillation. In this way, traces are 
transformed into skills, skills guide task generation, and generated  
tasks produce the traces needed for continued self-evolution.

Empirically, we validate \textsc{Socratic-SWE} on SWE-bench Verified, SWE-bench Lite, SWE-bench Pro, and Terminal-Bench 2.0. As summarized in Figure~\ref{fig:main_result}, Socratic-SWE consistently outperforms self-evolving baselines under the same compute budget, reaching 50.40\% on SWE-bench Verified after 3 iterations (+7.80 over the base agent and +3.40 over SSR). These results show that trace-derived curricula provide stronger training signals, support sustained self-evolution without external annotation, and transfer beyond repository repair to terminal-native tasks in realistic SWE environments. Our contributions are as follows:

\begin{enumerate}[leftmargin=*,topsep=0.5pt,itemsep=0pt,parsep=0pt]
    \item \textbf{Trace-Driven Self-Evolution Paradigm.}
    We show that solving traces, typically discarded after reward
    computation, can serve as a reusable substrate for self-evolving SWE
    agents. \textsc{Socratic-SWE} converts historical Solver behavior
    into structured agent skills, which guide subsequent
    repository-grounded task construction in a closed trace-skill-task
    loop.

    \item \textbf{Gradient-Aware Curriculum Optimization.}
    We introduce a curriculum mechanism that combines skill-conditioned
    task generation, execution-based task validation, and a
    solver-gradient alignment reward. Generated tasks are first verified
    in real repository environments to ensure reproducibility and
    solvability, and are then prioritized by whether their
    induced Solver updates align with trusted validation gradients.

    \item \textbf{Superior Empirical Performance.}
    We evaluate \textsc{Socratic-SWE} on SWE benchmarks, including
    SWE-bench and Terminal-Bench. Under the same compute
    budget, it consistently outperforms self-evolving baselines,
    demonstrating strong self-evolution and transfer capabilities.

\end{enumerate}

\begin{figure*}[t]
  \centering
  \includegraphics[width=0.95\textwidth]{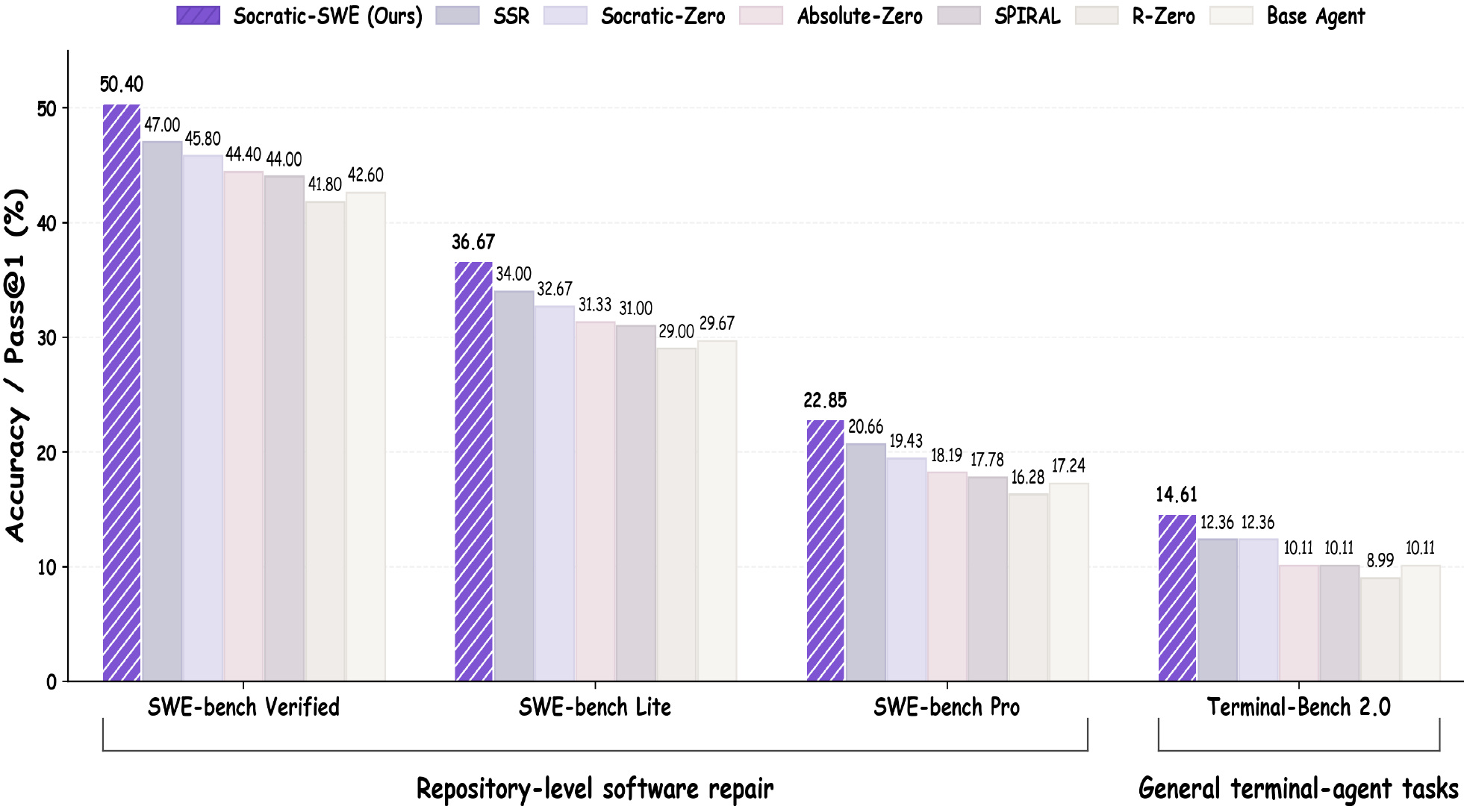}
\caption{
Main performance comparison on four benchmarks under a fixed 36k-instance training budget. Socratic-SWE achieves 50.40\% on SWE-bench Verified, 36.67\% on SWE-bench Lite, 22.85\% on SWE-bench Pro, and 14.61\% on Terminal-Bench 2.0, outperforming all adapted baselines across both repository-level software repair and terminal-agent tasks.
}
  \label{fig:main_result}
\end{figure*}

\section{Related Work}
\label{sec:related-work}

\noindent\textbf{Agentic Reinforcement Learning.}
Reinforcement learning is a central post-training paradigm for LLM reasoning and decision-making. OpenAI o1 popularized large-scale reasoning RL~\citep{openai2024o1card}, and DeepSeekMath's GRPO provided an open-source framework for verifiable rewards~\citep{shao2024deepseekmath}. DAPO~\citep{yu2026dapo}, GSPO~\citep{zheng2025gspo}, SAPO~\citep{gao2025sapo}, and GDPO~\citep{liu2026gdpo} improve RL efficiency and stability, but still assume immediately verifiable rewards. Agentic tasks violate this assumption: feedback is sparse and spread over long trajectories. GiGPO~\citep{feng2026gigpo}, ARPO~\citep{dong2026arpo}, and SkyRL-Agent~\citep{cao2025skyrlagent} extend RL to multi-turn tool use and long-horizon optimization. OPUS~\citep{wang2026opus}, GradAlign~\citep{yang2026gradalign}, and OptimSyn~\citep{fan2026optimsyn} use optimizer feedback, gradient alignment, or synthetic data construction as dynamic signals. Integrating verifiable feedback, failure analysis, and long-horizon agent training for SWE remains underexplored.

\noindent\textbf{SWE Coding Agents.}
SWE agents must navigate executable repositories, edit code, run commands, and preserve behavior under regression tests. SWE-bench introduced repository-level issue resolution from GitHub issues~\citep{jimenez2024swebench}. SWE-agent~\citep{yang2024sweagent}, OpenHands~\citep{wang2025openhands}, and Agentless~\citep{xia2025agentless} show that LLMs can solve tasks with tools, code execution, and verification feedback. SWE-Gym~\citep{pan2025swegym}, SWE-RL~\citep{wei2025swerl}, and SWE-Master~\citep{song2026swemaster} build training settings from executable environments, software evolution data, and post-training pipelines. BugLab~\citep{allamanis2021bugrepair}, SWE-smith~\citep{yang2025swesmith}, and SSR~\citep{wei2025selfplayswerl} expand data through bug injection, synthetic task generation, or self-play. Yet data construction emphasizes scale and executability, while using repository understanding, executable feedback, and model failures only weakly for targeted task generation.

\noindent\textbf{Self-Evolving LLMs.}
Self-play reduces human labels by assigning proposer, solver, evaluator, or teacher roles. TTRL~\citep{zuo2025ttrl}, R-Zero~\citep{huang2026rzero}, and Socratic-Zero~\citep{wang2025socraticzero} form self-improvement signals from unlabeled data, Challenger-Solver co-evolution, or Teacher-Solver-Generator loops. Absolute Zero~\citep{zhao2025absolutezero}, SPIRAL~\citep{liu2025spiral}, Socratic-Geo~\citep{jiao2026socraticgeo}, and SpatialEvo~\citep{li2026spatialevo} reduce noisy feedback through verifiable environments, including code execution, zero-sum games, programmatic geometry, and deterministic spatial tasks. SkillRL~\citep{xia2026skillrl} and SKILL0~\citep{lu2026skill0} distill interaction experience into reusable skills. Agentic Proposing~\citep{jiao2026agenticproposing} uses modular reasoning skills to synthesize harder verifiable examples. Overall, self-evolution is shifting toward environment-grounded feedback and skill-driven synthesis; for SWE, the key challenge is transforming long-horizon repair failures into future task distributions that keep pushing model capabilities.

\section{Methodology}
\label{sec:method}

\subsection{Socratic-SWE Framework}
\label{sec:method_overview}

We introduce \textbf{Socratic-SWE}, a co-evolutionary self-play framework for software engineering agents. A shared policy $\pi_\theta$ alternates between two roles: a \textbf{Generator} that constructs repository-grounded repair tasks and a \textbf{Solver} that produces patches for them. The framework learns from two external signals: an \textbf{Agent Skill Registry} $\mathcal{S}$ distilled from historical interaction traces, and a staged execution-grounded validation pipeline that filters generated tasks before they enter training. Figure~\ref{fig:method} illustrates the overall loop.

\begin{figure*}[t]
  \centering
  \includegraphics[width=0.95\textwidth]{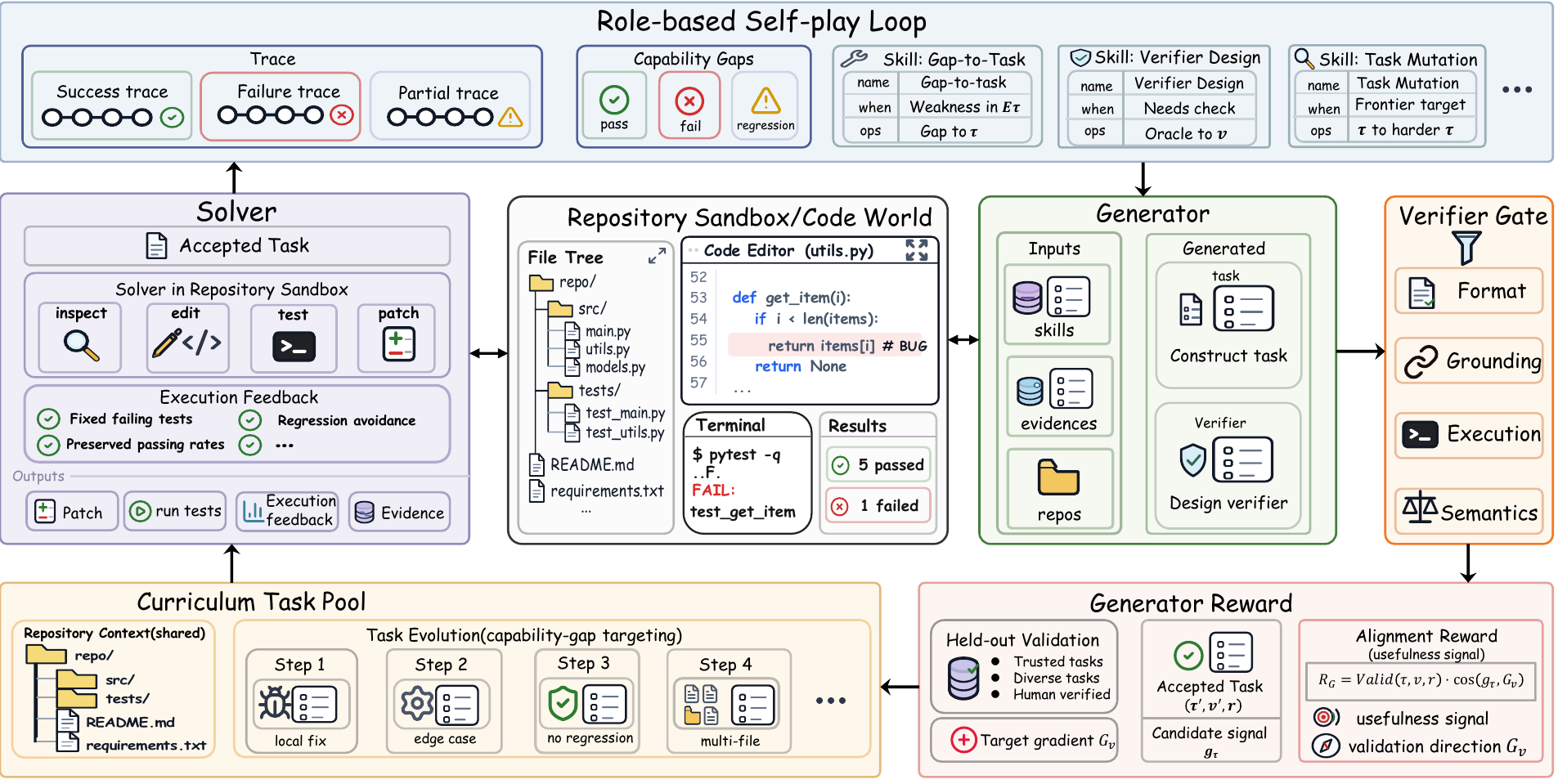}
\caption{
An overview of the Socratic-SWE framework, illustrating the co-evolution of the Solver and the Generator. \textbf{On the Solver side}, the model consumes tasks from the task pool, solves them in the repository sandbox, and produces solving traces. These traces are distilled into the Agent Skill Registry, where recurring model gaps and repository-understanding patterns are captured as reusable skills. \textbf{On the Generator side}, the model uses the Skill Registry to synthesize targeted tasks in the repository sandbox. Candidate tasks are then filtered by the Verifier Gate along format, grounding, execution, and semantics to ensure that they are well-formed, reproducible, solvable, and meaningful. Accepted tasks receive an alignment reward from the held-out validation gradient, enter the task pool as the next generation of tasks, and are consumed again by the Solver, forming a closed task-trace-skill-task loop.
}
  \label{fig:method}
\end{figure*}

Let $\mathcal{R}$ denote the repository corpus and $\mathcal{S}$ the Agent Skill Registry. At iteration $t$, the system maintains a curriculum of tasks $\mathcal{D}_t = \{(r, \tau, v)\}$, where $r \in \mathcal{R}$ is a sandboxed repository, $\tau$ is a repair task grounded in $r$, and $v$ is an executable verification signal. The Solver trains on $\mathcal{D}_t$ and produces trajectories that expose capability gaps. The Generator then uses skills $s \in \mathcal{S}$ to propose tasks targeting these gaps, and the validation pipeline filters proposals so that only executable, reproducible, and solvable tasks enter the curriculum:
\begin{align}
(\tau', v') &\sim \pi_\theta(\cdot \mid r, s, \mathcal{E}_t, \mathrm{role}{=}G), \\
\label{eq:curriculum_update}
\mathcal{D}_{t+1} &= \mathcal{D}_t \cup \{(r, \tau', v') \mid \mathcal{V}_{\mathrm{alid}}(\tau', v', r) = 1\},
\end{align}
where $\mathcal{E}_t$ denotes Solver-side evidence collected on $\mathcal{D}_t$, and $\mathcal{V}_{\mathrm{alid}}(\cdot)$ is a staged validation function defined in \S\ref{sec:method_generator}. Algorithm~\ref{alg:socratic-swe-training} details the procedure.

\subsection{Agent Skill Registry}
\label{sec:method_skills}

A central challenge in skill-guided task generation is obtaining skills that are both structured for retrieval and grounded in real agent behavior. We distill an \textbf{Agent Skill Registry} from historical interaction traces through a three-stage pipeline.

We define a \emph{skill} as a structured document with four fields: a \emph{name}, a natural-language \emph{description}, a set of \emph{applicability conditions}, and an ordered list of \emph{operations}. This representation lets the Generator retrieve and condition on skills programmatically rather than through free-form text.

\textbf{Stage 1: Trace Collection.}
We deploy $\pi_\theta$ at its current checkpoint on the seed task set and collect a trace corpus $\mathcal{T} = \{\tau_1, \ldots, \tau_N\}$. Each trace records repository inspection, code edits, command execution, and verification outcomes. We split $\mathcal{T}$ into successful traces $\mathcal{T}^{+} = \{\tau : r(\tau) = 1\}$ and failed traces $\mathcal{T}^{-} = \{\tau : r(\tau) = 0\}$, and retain both since failure traces expose capability gaps.

\textbf{Stage 2: Skill Extraction.}
A distillation model $\mathcal{M}_{\mathrm{distill}}$ processes the trace corpus and extracts recurring behavioral patterns as candidate skills. For successful traces, it identifies generalizable strategies; for failed traces, it summarizes failure lessons and corrective principles. Formally:
\begin{equation}
\label{eq:skill_extract}
\hat{\mathcal{S}} = \mathcal{M}_{\mathrm{distill}}\!\big(\mathcal{T}^{+},\, \mathcal{T}^{-}\big),
\end{equation}
where $\hat{\mathcal{S}}$ is the set of candidate skills.

\textbf{Stage 3: Registry Construction.}
Candidate skills are deduplicated by semantic similarity and filtered by trace coverage to form the Skill Registry:
\begin{equation}
\label{eq:registry}
\mathcal{S} = \mathrm{Dedup}\!\big(\hat{\mathcal{S}},\, \delta_{\mathrm{sim}}\big) = \{s_1, \ldots, s_M\},
\end{equation}
where $\delta_{\mathrm{sim}}$ is the similarity threshold and $M$ is the number of retained skills. During training, the Generator samples $s \sim \mathcal{S}$ and uses it to bias task proposal toward the behavioral pattern.
\subsection{Skill-Guided Task Generator}
\label{sec:method_generator}

When acting as Generator, $\pi_\theta$ constructs executable SWE repair tasks that expose Solver weaknesses. Given a repository $r \in \mathcal{R}$ and a skill $s \in \mathcal{S}$, the Generator produces a candidate task $\tau$ and its verification signal $v$:
\begin{equation}
\label{eq:generator_proposal}
(\tau, v) \sim \pi_\theta(\cdot \mid r, s, \mathcal{E}_t, \mathrm{role}{=}G).
\end{equation}
Here, $\tau$ specifies a repository-grounded repair objective and $v$ specifies the executable tests or commands used to evaluate patches. Conditioning on Solver evidence $\mathcal{E}_t$ makes generation adaptive as the Solver improves.

\textbf{Task Validation.}
Not all generated tasks are suitable for training. Before entering the curriculum, each candidate $(\tau, v)$ passes four checks in the repository sandbox $r$: format, grounding, execution, and semantics. A candidate is accepted iff all stages pass:
\begin{enumerate}
\item \emph{Format} ($f_1$): $\tau$ and $v$ are well-formed, parseable, and syntactically valid.
\item \emph{Grounding} ($f_2$): $\tau$ references artifacts that actually exist in $r$.
\item \emph{Execution} ($f_3$): $v$ runs without infrastructure errors and is stable across repeated runs.
\item \emph{Semantics} ($f_4$): $v$ separates failing from repaired states, and at least one valid repair exists.
\end{enumerate}
We write:
\begin{equation}
\label{eq:validation}
\mathcal{V}_{\mathrm{alid}}(\tau, v, r) = \prod_{l=1}^{4} f_l(\tau, v, r) \;\in\; \{0, 1\},
\end{equation}
where each $f_l$ is evaluated only if all preceding stages pass. Only accepted candidates enter $\mathcal{D}_{t+1}$.

\textbf{Generator Reward.}
Validation ensures a task is executable and solvable, but not useful. We reward the Generator by whether the Solver update by a task aligns with a validation direction.

We maintain a held-out set of trusted validation tasks $\mathcal{V}_{\text{val}} = \{(\tau_j^v, v_j^v, r_j^v)\}_{j=1}^{|\mathcal{V}_{\text{val}}|}$. For each validation task, we roll out $K$ Solver trajectories, compute executable feedback rewards, and estimate the per-task policy gradient:
\begin{equation}
\label{eq:val_grad}
g_j^v = \tfrac{1}{K}\textstyle\sum_{k=1}^{K} \hat{A}_{j,k}\, \nabla_\theta \log \pi_\theta(\hat{y}_{j,k} \mid \tau_j^v, r_j^v, v_j^v).
\end{equation}
Averaging over the validation set yields the target gradient direction
$G_v = \frac{1}{|\mathcal{V}_{\text{val}}|}\sum_j g_j^v$.

For each candidate task $(\tau, v)$ in repository $r$, we estimate the Solver policy gradient $g_\tau$ from $K$ rollouts. The Generator reward is:
\begin{equation}
\label{eq:generator_reward}
R_G(\tau, v, r) = \mathcal{V}_{\mathrm{alid}}(\tau, v, r) \cdot \cos(g_\tau,\, G_v).
\end{equation}
The validation factor zeros out invalid tasks, and the cosine term favors tasks whose induced updates align with the validation gradient. We recompute $G_v$ periodically as the Solver evolves.

\subsection{Repository Repair Solver}
\label{sec:method_solver}

When acting as Solver, $\pi_\theta$ solves accepted tasks by interacting with the repository environment. The Solver sees the task specification and repository feedback, but not the reference solution or verifier internals used by the Generator.

\textbf{Patch Generation.}
Given a task $(\tau, v)$ and repository $r$, the Solver samples a trajectory:
\begin{equation}
\label{eq:solver_generation}
\hat{y} \sim \pi_\theta(\cdot \mid \tau, r, v, \mathrm{role}{=}S).
\end{equation}
The trajectory may include inspection, code localization, file edits, and validation attempts.

\textbf{Executable Feedback Reward.}
For each trajectory $\hat{y}$, we apply the generated patch and run the verification suite. Let $\mathcal{F}$ denote the set of originally failing tests and $\mathcal{P}$ the set of originally passing tests; let $\mathcal{F}_{\checkmark} \subseteq \mathcal{F}$ and $\mathcal{P}_{\checkmark} \subseteq \mathcal{P}$ be the subsets that pass after patching. The Solver reward is:
\begin{equation}
\label{eq:solver_reward}
r_S = \lambda_1\, \mathbf{1}\!\left[\mathcal{F}_{\checkmark} {=} \mathcal{F} \wedge \mathcal{P}_{\checkmark} {=} \mathcal{P}\right] + \lambda_2\, \frac{|\mathcal{F}_{\checkmark}|}{|\mathcal{F}|} + \lambda_3\, \frac{|\mathcal{P}_{\checkmark}|}{|\mathcal{P}|},
\end{equation}
where the three terms reward full-suite pass, partial repair rate, and regression avoidance.

\subsection{Training with Role-Specific Objectives}
\label{sec:method_training}

We optimize both roles jointly with shared weights:
\begin{align}
\label{eq:joint_objective}
J(\theta) = \;&\mathbb{E}_{r,s}\!\big[\mathbb{E}_{(\tau,v) \sim \pi_G}\![R_G(\tau,v,r)]\big] \notag \\
+\;&\mathbb{E}_{\tau,r,v}\!\big[\mathbb{E}_{\hat{y} \sim \pi_S}\![r_S(\hat{y},\tau,v,r)]\big],
\end{align}
where $\pi_G(\cdot \mid r, s) := \pi_\theta(\cdot \mid r, s, \mathrm{role}{=}G)$ and
$\pi_S(\cdot \mid \tau, r, v) := \pi_\theta(\cdot \mid \tau, r, v, \mathrm{role}{=}S)$.

We optimize a clipped surrogate objective:
\begin{equation}
\label{eq:grpo_objective}
\mathcal{L}(\theta) = \mathbb{E}\!\left[\tfrac{1}{K}\textstyle\sum_{i=1}^{K} \hat{\mathcal{L}}_i - \beta\, D_{\mathrm{KL}}[\pi_\theta \,\|\, \pi_{\mathrm{ref}}]\right],
\end{equation}
where $\hat{\mathcal{L}}_i = \min\!\big(\rho_i \hat{A}_i,\; \mathrm{clip}(\rho_i, 1{-}\epsilon, 1{+}\epsilon)\,\hat{A}_i\big)$ and
$\rho_i = \frac{\pi_\theta(o_i \mid x)}{\pi_{\theta_{\mathrm{old}}}(o_i \mid x)}$ is the importance ratio.

For the Generator, which receives a single scalar reward $R_G$ (Eq.~\eqref{eq:generator_reward}), we apply GRPO with $R_i := R_G^i$. The Solver reward $r_S$ (Eq.~\eqref{eq:solver_reward}) combines three heterogeneous components ($\mathit{pass}$, $\mathit{repair}$, $\mathit{regr}$) on different scales. We adopt GDPO~\citep{liu2026gdpo}, which normalizes each component within its own group before aggregation:
\begin{equation}
\label{eq:gdpo_per_reward}
\hat{A}^{(m)}_i = \frac{r_S^{(m,i)} - \mathrm{mean}(\{r_S^{(m,j)}\}_{j=1}^K)}{\mathrm{std}(\{r_S^{(m,j)}\}_{j=1}^K) + \delta},
\end{equation}
\begin{equation}
\label{eq:gdpo_objective}
\hat{A}_{S}^{i} = \mathrm{BatchNorm}\!\Big(\textstyle\sum_{m=1}^{3} \hat{A}^{(m)}_i\Big),
\end{equation}
where $m \in \{1,2,3\}$ indexes reward components and $\mathrm{BatchNorm}$ rescales advantages across the batch. This normalization lets the Solver distinguish partial repair from full fix and yields informative gradients from a single summed reward. Both roles share the clipped objective in Eq.~\eqref{eq:grpo_objective}, with $\hat{A}_i := \hat{A}_{S}^{i}$ for the Solver. The total loss is $\mathcal{L}(\theta) = \mathcal{L}_G(\theta) + \mathcal{L}_S(\theta)$.
\section{Experiments}
\label{sec:experiments}

\subsection{Experimental Setup}
\newcommand{\mnum}[1]{\makebox[5.2em][l]{#1}}
\newcommand{\mdelta}[1]{\makebox[4.2em][l]{#1}}

\begin{table*}[!t]
\centering
\footnotesize
\setlength{\tabcolsep}{3.5pt}
\renewcommand{\arraystretch}{0.88}
\arrayrulecolor{RuleGray}

\caption{
Main results on software engineering benchmarks.
All self-evolving methods use identical Solver (Qwen3.5-9B), harness (mini-swe-agent), and training budget (12k instances $\times$ 3 iterations).
TB2 is evaluated with little-coder.
Overall is the mean score across four benchmarks.
Green and red arrows indicate improvements and degradations relative to the Base Agent, respectively.
}
\label{tab:main_results}

\resizebox{\textwidth}{!}{
\begin{tabular}{l|c|cccc|c}
\specialrule{1.1pt}{0pt}{0pt}
\multicolumn{1}{c|}{} &
\multicolumn{1}{c|}{} &
\multicolumn{4}{c|}{\textbf{Benchmarks}} &
\multicolumn{1}{c}{} \\
\cmidrule(lr){3-6}
\textbf{Method} &
\cellcolor{OverallGray}\textbf{Overall} &
\makecell{\textbf{SWE-bench}\\\textbf{Verified}} &
\makecell{\textbf{SWE-bench}\\\textbf{Lite}} &
\makecell{\textbf{SWE-bench}\\\textbf{Pro}} &
\makecell{\textbf{Terminal-Bench}\\\textbf{2.0}} &
\makecell{\textbf{$\Delta$ vs.}\\\textbf{Base}} \\
\midrule

Qwen3.5-9B & & & & & & \\

\quad + Base Agent & \mnum{\ogray{24.91}} & \mnum{42.60} & \mnum{29.67} & \mnum{17.24} & \mnum{10.11} & \mdelta{---} \\

\midrule
\rowcolor{BlockLavender}
\emph{R-Zero} & & & & & & \\
{\scriptsize\quad + Iteration 1} & \mnum{\ogray{\scriptsize 25.29}} & \mnum{\scriptsize 43.20\uar{0.60}} & \mnum{\scriptsize 30.33\uar{0.66}} & \mnum{\scriptsize 17.51\uar{0.27}} & \mnum{\scriptsize 10.11\nodiff} & \mdelta{\scriptsize \uar{0.38}} \\
{\scriptsize\quad + Iteration 2} & \mnum{\ogray{\scriptsize 24.25}} & \mnum{\scriptsize 42.00\dar{0.60}} & \mnum{\scriptsize 29.33\dar{0.34}} & \mnum{\scriptsize 16.69\dar{0.55}} & \mnum{\scriptsize 8.99\dar{1.12}} & \mdelta{\scriptsize \dar{0.66}} \\
{\scriptsize\quad + Iteration 3} & \mnum{\ogray{\scriptsize 24.02}} & \mnum{\scriptsize 41.80\dar{0.80}} & \mnum{\scriptsize 29.00\dar{0.67}} & \mnum{\scriptsize 16.28\dar{0.96}} & \mnum{\scriptsize 8.99\dar{1.12}} & \mdelta{\scriptsize \dar{0.89}} \\

\midrule
\rowcolor{BlockLavender}
\emph{SPIRAL} & & & & & & \\
{\scriptsize\quad + Iteration 1} & \mnum{\ogray{\scriptsize 25.67}} & \mnum{\scriptsize 43.80\uar{1.20}} & \mnum{\scriptsize 31.00\uar{1.33}} & \mnum{\scriptsize 17.78\uar{0.54}} & \mnum{\scriptsize 10.11\nodiff} & \mdelta{\scriptsize \uar{0.76}} \\
{\scriptsize\quad + Iteration 2} & \mnum{\ogray{\scriptsize 26.29}} & \mnum{\scriptsize 44.40\uar{1.80}} & \mnum{\scriptsize 31.33\uar{1.66}} & \mnum{\scriptsize 18.19\uar{0.95}} & \mnum{\scriptsize 11.24\uar{1.13}} & \mdelta{\scriptsize \uar{1.38}} \\
{\scriptsize\quad + Iteration 3} & \mnum{\ogray{\scriptsize 25.72}} & \mnum{\scriptsize 44.00\uar{1.40}} & \mnum{\scriptsize 31.00\uar{1.33}} & \mnum{\scriptsize 17.78\uar{0.54}} & \mnum{\scriptsize 10.11\nodiff} & \mdelta{\scriptsize \uar{0.81}} \\

\midrule
\rowcolor{BlockLavender}
\emph{Absolute-Zero} & & & & & & \\
{\scriptsize\quad + Iteration 1} & \mnum{\ogray{\scriptsize 26.21}} & \mnum{\scriptsize 44.20\uar{1.60}} & \mnum{\scriptsize 31.33\uar{1.66}} & \mnum{\scriptsize 18.06\uar{0.82}} & \mnum{\scriptsize 11.24\uar{1.13}} & \mdelta{\scriptsize \uar{1.30}} \\
{\scriptsize\quad + Iteration 2} & \mnum{\ogray{\scriptsize 26.75}} & \mnum{\scriptsize 45.00\uar{2.40}} & \mnum{\scriptsize 32.00\uar{2.33}} & \mnum{\scriptsize 18.74\uar{1.50}} & \mnum{\scriptsize 11.24\uar{1.13}} & \mdelta{\scriptsize \uar{1.84}} \\
{\scriptsize\quad + Iteration 3} & \mnum{\ogray{\scriptsize 26.01}} & \mnum{\scriptsize 44.40\uar{1.80}} & \mnum{\scriptsize 31.33\uar{1.66}} & \mnum{\scriptsize 18.19\uar{0.95}} & \mnum{\scriptsize 10.11\nodiff} & \mdelta{\scriptsize \uar{1.10}} \\

\midrule
\rowcolor{BlockLavender}
\emph{Socratic-Zero} & & & & & & \\
{\scriptsize\quad + Iteration 1} & \mnum{\ogray{\scriptsize 26.83}} & \mnum{\scriptsize 45.20\uar{2.60}} & \mnum{\scriptsize 32.00\uar{2.33}} & \mnum{\scriptsize 18.88\uar{1.64}} & \mnum{\scriptsize 11.24\uar{1.13}} & \mdelta{\scriptsize \uar{1.92}} \\
{\scriptsize\quad + Iteration 2} & \mnum{\ogray{\scriptsize 27.87}} & \mnum{\scriptsize 46.40\uar{3.80}} & \mnum{\scriptsize 33.00\uar{3.33}} & \mnum{\scriptsize 19.70\uar{2.46}} & \mnum{\scriptsize 12.36\uar{2.25}} & \mdelta{\scriptsize \uar{2.96}} \\
{\scriptsize\quad + Iteration 3} & \mnum{\ogray{\scriptsize 27.57}} & \mnum{\scriptsize 45.80\uar{3.20}} & \mnum{\scriptsize 32.67\uar{3.00}} & \mnum{\scriptsize 19.43\uar{2.19}} & \mnum{\scriptsize 12.36\uar{2.25}} & \mdelta{\scriptsize \uar{2.66}} \\

\midrule
\rowcolor{BlockLavender}
\emph{SSR} & & & & & & \\
{\scriptsize\quad + Iteration 1} & \mnum{\ogray{\scriptsize 26.31}} & \mnum{\scriptsize 44.20\uar{1.60}} & \mnum{\scriptsize 31.33\uar{1.66}} & \mnum{\scriptsize 18.47\uar{1.23}} & \mnum{\scriptsize 11.24\uar{1.13}} & \mdelta{\scriptsize \uar{1.40}} \\
{\scriptsize\quad + Iteration 2} & \mnum{\ogray{\scriptsize 27.63}} & \mnum{\scriptsize 45.80\uar{3.20}} & \mnum{\scriptsize 32.67\uar{3.00}} & \mnum{\scriptsize 19.70\uar{2.46}} & \mnum{\scriptsize 12.36\uar{2.25}} & \mdelta{\scriptsize \uar{2.72}} \\
{\scriptsize\quad + Iteration 3} & \mnum{\ogray{\scriptsize 28.51}} & \mnum{\scriptsize 47.00\uar{4.40}} & \mnum{\scriptsize 34.00\uar{4.33}} & \mnum{\scriptsize 20.66\uar{3.42}} & \mnum{\scriptsize 12.36\uar{2.25}} & \mdelta{\scriptsize \uar{3.60}} \\

\midrule
\rowcolor{BlockLavender}
\emph{Socratic-SWE (Ours)} & & & & & & \\
{\scriptsize\quad + Iteration 1} & \mnum{\ogray{\scriptsize 27.82}} & \mnum{\scriptsize 46.20\uar{3.60}} & \mnum{\scriptsize 33.00\uar{3.33}} & \mnum{\scriptsize 19.70\uar{2.46}} & \mnum{\scriptsize 12.36\uar{2.25}} & \mdelta{\scriptsize \uar{2.91}} \\
{\scriptsize\quad + Iteration 2} & \mnum{\ogray{\scriptsize 29.64}} & \mnum{\scriptsize 48.40\uar{5.80}} & \mnum{\scriptsize 35.33\uar{5.66}} & \mnum{\scriptsize 21.34\uar{4.10}} & \mnum{\scriptsize 13.48\uar{3.37}} & \mdelta{\scriptsize \uar{4.73}} \\
\rowcolor{OursPink}
{\scriptsize\quad + \textbf{Iteration 3}} &
\cellcolor{OursPink}{\mnum{\scriptsize\textbf{31.13}}} &
\mnum{\scriptsize\textbf{50.40\uar{7.80}}} &
\mnum{\scriptsize\textbf{36.67\uar{7.00}}} &
\mnum{\scriptsize\textbf{22.85\uar{5.61}}} &
\mnum{\scriptsize\textbf{14.61\uar{4.50}}} &
\mdelta{\scriptsize\textbf{\uar{6.22}}} \\

\specialrule{1.1pt}{0pt}{0pt}
\end{tabular}
}
\end{table*}
\noindent\textbf{Models.}\hspace{0.35em}
Qwen3.5-9B \cite{qwen2026qwen35} was used for both the Generator and the Solver. The Generator was optimized with GRPO~\cite{shao2024deepseekmath} from gradient feedback. Qwen3.6-27B~\cite{qwen2026qwen3627b} was used to distill skills. The Solver was trained with GDPO~\cite{liu2026gdpo} on data generated by Socratic-SWE, with weights shared with the Generator.

\noindent\textbf{Datasets and Benchmarks.}\hspace{0.35em}
All self-evolving methods ran for 3 iterations, generating 12k validated training instances per iteration (36k total). For baselines that require seed tasks, 10\% of SWE-smith~\cite{yang2025swesmith} was used as the seed dataset; Socratic-SWE requires seed repositories. The trusted validation set $\mathcal{V}_{\mathrm{val}}$ was a fixed held-out subset of BeyondSWE~\cite{chen2026beyondswe}, used for generator-gradient alignment. All methods were evaluated on four benchmarks: SWE-bench Verified~\cite{jimenez2024swebench} (500 human-verified repository-level issues), SWE-bench Lite~\cite{jimenez2024swebench} (300 filtered issues), SWE-bench Pro Public~\cite{deng2025swebenchpro} (731 complex enterprise-level software problems), and Terminal-Bench 2.0 (TB2)~\cite{merrill2026terminalbench} (terminal-native task completion in sandboxed environments).

\noindent\textbf{Agent Harness.}\hspace{0.35em}
mini-swe-agent~\cite{sweagent2025minisweagent} was used as the execution harness for SWE benchmarks, with only Bash exposed to reduce confounds from tool design. For Terminal-Bench 2.0, we used little-coder~\cite{inbar2026littlecoder} as the terminal-agent evaluation harness.

\noindent\textbf{Evaluation.}\hspace{0.35em}
For SWE-bench Verified, Lite, and Pro, each agent interacted with the repository and submitted a patch validated by benchmark tests. For TB2, each agent was given an instruction and a sandboxed terminal, then required to reach a final state that passed the verifier. TB2 was evaluated using little-coder~\cite{inbar2026littlecoder}. Pass rate was reported as the primary metric across all benchmarks.

\noindent\textbf{Baseline Methods.}\hspace{0.35em}
We compare Socratic-SWE with the base agent and five self-evolving baselines: SPIRAL~\cite{liu2025spiral}, R-Zero~\cite{huang2026rzero}, Absolute-Zero~\cite{zhao2025absolutezero}, Socratic-Zero~\cite{wang2025socraticzero}, and SSR~\cite{wei2025selfplayswerl}. All methods use the same Solver architecture (Qwen3.5-9B), agent harness (mini-swe-agent), Terminal-Bench harness (little-coder), benchmarks, and interaction budget. Details of the SWE adaptations and baseline initialization are provided in Appendix~\ref{app:baseline_details}.

\subsection{Main Results}
\label{sec:main_results}

\noindent\textbf{Overall Performance.}\hspace{0.35em}
As shown in Table~\ref{tab:main_results}, Socratic-SWE achieves the strongest results across all four benchmarks after 3 iterations. With a fixed 36k-instance budget and \emph{zero} pre-existing SWE training tasks, it improves over the base agent by +6.22 points overall, reaching 50.40\% on SWE-bench Verified (+7.80), 36.67\% on Lite (+7.00), 22.85\% on Pro (+5.61), and 14.61\% on Terminal-Bench 2.0 (+4.50). Its Verified gain also grows steadily across iterations (+3.60, +5.80, +7.80), while several baselines saturate or regress.

\noindent\textbf{Self-Play Methods Struggle on SWE Tasks.}\hspace{0.35em}
General self-play methods are brittle when adapted to SWE. R-Zero, which uses majority vote as reward, briefly improves at Iteration~1 (+0.60 on Verified) but falls below the base agent ($-$0.80), suggesting that vote-based rewards are too noisy for partial repairs. SPIRAL and Absolute-Zero peak at Iteration~2 (+1.80 and +2.40 on Verified) and regress at Iteration~3, indicating that self-play without execution-grounded validation can lead to mode collapse or reward hacking.

\noindent\textbf{SSR Improves Steadily but Remains Bounded.}\hspace{0.35em}
SSR is the strongest baseline, reaching +4.40 on SWE-bench Verified after 3 iterations. Its bug-injection mechanism and execution checks keep generated tasks valid, but its gains lack the acceleration of Socratic-SWE. This suggests that, without skill-guided targeting of capability gaps, SSR exhausts low-hanging bug patterns.

\noindent\textbf{Teacher-Guided Co-evolution Saturates Early.}\hspace{0.35em}
Socratic-Zero, despite using a 397B-parameter Teacher, peaks at Iteration~2 (+3.80 on Verified) and slightly declines at Iteration~3. This suggests that teacher-guided task generation remains bounded by the Teacher's domain understanding and may lose information during Generator distillation.

\noindent\textbf{Terminal-Native Generalization.}\hspace{0.35em}
On Terminal-Bench 2.0, only SSR and Socratic-SWE show meaningful gains (+2.25 and +4.50). Socratic-SWE's stronger transfer suggests that trace-derived skills capture general agent behaviors across tasks such as file manipulation, command chaining, and environment navigation, rather than only issue-specific repair patterns.


\section{Analysis and Ablation Studies}
\label{sec:analysis}

\subsection{Iteration Scaling and Saturation}
\label{sec:iteration_scaling}

To study long-term scaling, we extend training to 5 iterations for Socratic-SWE and SSR. Figure~\ref{fig:iteration_scaling} shows three regimes: Socratic-SWE improves rapidly in Iterations 1--3 (+3.60, +2.20, +2.00) and still gains in Iteration 4 (+1.20, reaching 51.60\%), while SSR improves more slowly (+1.60, +1.60, +1.20, then +0.80). By Iteration 5, both methods nearly plateau, at 52.00\% for Socratic-SWE and 48.00\% for SSR.

\begin{figure}[t]
  \centering
  \includegraphics[width=\linewidth]{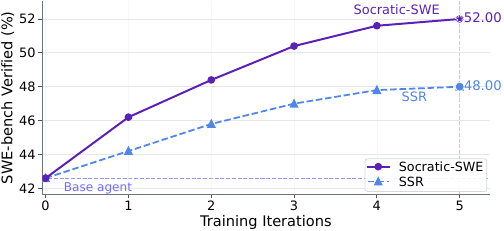}
\caption{
SWE-bench Verified scaling beyond the first three iterations. Socratic-SWE improves faster and reaches a higher ceiling, while SSR plateaus earlier.
}
  \label{fig:iteration_scaling}
\end{figure}

This saturation reflects the closed-world setup. The Agent Skill Registry increasingly covers the seed-repository space, leaving fewer gaps to target, while the fixed repository pool yields redundant training signal. Even so, Socratic-SWE saturates 2 iterations later and at a higher ceiling than SSR, showing that skill-guided curriculum generation extracts more signal from the same data. Expanding the repository pool or enabling cross-repository transfer may extend this trend.

\subsection{Ablation of Framework Components}
\label{sec:ablation_components}

We ablate three components on SWE-bench Verified at Iteration 3 (Table~\ref{tab:ablation}): removing the Skill Registry, replacing trace distillation with manual skills, and replacing GDPO with GRPO.

\begin{table}[t]
\centering
\footnotesize
\setlength{\tabcolsep}{4.5pt}
\renewcommand{\arraystretch}{0.98}
\arrayrulecolor{RuleGray}

\caption{
Ablation results on SWE-bench Verified at Iteration 3. The upper block removes key Socratic-SWE components, and the lower block varies the skill extractor. The results show that each component helps and that performance is robust across extractor scales.
}
\label{tab:ablation}

\begin{tabular}{lrr}
\specialrule{1.1pt}{0pt}{0pt}
\textbf{Variant} & \textbf{Verified} & \textbf{$\Delta$ vs. Full} \\
\midrule

\rowcolor{BlockLavender}
\textbf{Socratic-SWE} & \vnum{\textbf{50.40}}\varr{\phantom{\dar{0.60}}} & \varr{\textbf{--}} \\
\quad w/o Skill Registry & \vnum{46.20}\varr{\dar{4.20}} & \varr{\dar{4.20}} \\
\quad w/o Trace Distillation & \vnum{48.00}\varr{\dar{2.40}} & \varr{\dar{2.40}} \\
\quad w/o GDPO, using GRPO & \vnum{48.60}\varr{\dar{1.80}} & \varr{\dar{1.80}} \\

\midrule
\multicolumn{3}{l}{\textit{Skill extraction model}} \\
\quad Qwen3.5-9B self-extraction & \vnum{49.80}\varr{\dar{0.60}} & \varr{\dar{0.60}} \\
\quad Qwen3.6-27B default & \vnum{50.40}\varr{\nodiff} & \varr{\nodiff} \\
\quad Claude Opus 4.5 & \vnum{51.00}\varr{\uar{0.60}} & \varr{\uar{0.60}} \\
\bottomrule
\end{tabular}
\end{table}

Removing the Skill Registry causes the largest drop, showing that curriculum design is the main driver of Socratic-SWE. Replacing trace distillation with manual skills hurts performance because hand-written taxonomies miss behaviors such as iterative grep-then-edit patterns. Replacing GDPO with GRPO also degrades results, suggesting that decomposed reward learning better handles partially correct patches than binary pass/fail signals.

\subsection{Robustness of Skill Extraction}
\label{sec:ablation_extraction}

Socratic-SWE is not sensitive to the skill extractor. Even the smallest extractor, Qwen3.5-9B, reaches 49.80\% on Verified, only 0.60 points below the full system. Qwen3.6-27B gives a small gain, and Claude Opus 4.5 adds another +0.60. This suggests that the core advantage comes from the framework itself---skill-gap-targeted generation plus execution-grounded validation---rather than from a particularly strong extractor. As long as skills come from real traces and are validated by tests, coarse descriptions are sufficient to guide effective curriculum generation.
\section{Conclusion}
\label{sec:conclusion}

We presented \textsc{Socratic-SWE}, a practical closed-loop framework for self-evolving software engineering agents under limited access to high-quality SWE task data. By reusing historical solving traces as training signal and distilling them into an Agent Skill Registry, \textsc{Socratic-SWE} generates targeted repair tasks that address capability gaps and track the model's frontier. Across four benchmarks, it gains +7.80 points on SWE-bench Verified and +4.50 points on Terminal-Bench 2.0 after three iterations, outperforming five self-play baselines under identical compute budgets. Analyses show that skill-guided curricula delay saturation and better exploit finite repository pools, while ablations confirm the roles of the skill registry, trace distillation, and GDPO. Future work may explore dynamic repository augmentation, cross-domain skill transfer, and online skill discovery. 
\section*{Limitations}
\label{sec:limitations}

Socratic-SWE is evaluated in a closed-world setting with a fixed pool of seed repositories. As the Agent Skill Registry becomes more complete, later iterations have fewer novel capability gaps to target, which makes task generation increasingly redundant. Our scaling analysis therefore reflects curriculum behavior under a fixed repository distribution rather than fully open-ended improvement with continuously refreshed data.

The method also relies on a held-out validation set to define the generator-gradient alignment reward. While this design improves curriculum quality, it introduces dependence on the choice of trusted validation tasks and may limit robustness if the validation distribution is not representative of the target deployment setting.

In addition, Socratic-SWE assumes executable verification and sandboxed repository interaction. Its gains may therefore not transfer directly to settings without reliable tests, deterministic execution, or clear task-level validation. Finally, our evaluation covers four benchmarks in SWE and terminal-agent settings; broader transfer to other programming languages, repository styles, and software engineering workflows remains to be established.

\clearpage

\bibliographystyle{unsrt}
\bibliography{references}

\clearpage

\beginappendix

\section{Training Algorithm}
\label{app:training_algorithm}

Algorithm~\ref{alg:socratic-swe-training} summarizes the role-conditioned
Generator--Solver training loop used in Socratic-SWE.

\begin{algorithm*}[t]
\caption{Socratic-SWE: Self-Play in Repository Environments}
\label{alg:socratic-swe-training}
\begin{algorithmic}[1]
\Require Shared policy $\pi_\theta$; repository corpus $\mathcal{R}$; Agent Skill Registry $\mathcal{S}$; curriculum $\mathcal{D}_0$; batch size $B$; group size $K$; iterations $T$; held-out validation set $\mathcal{V}$
\For{$t \gets 1$ to $T$}
  \tikzmk{A}\State \textbf{Generator Phase:} Construct skill-guided SWE repair tasks
  \State $\mathcal{E}_t \gets$ Solver outcomes on $\mathcal{D}_t$
  \For{$b \gets 1$ to $B$}
    \State Sample $r \sim \mathcal{R}$, retrieve skill $s \sim \mathcal{S}$
    \State $\{(\tau_k, v_k)\}_{k=1}^{K} \gets \pi_\theta(\cdot \mid r, s, \mathcal{E}_t, \mathrm{role}{=}G)$ \Comment{$K$ candidate tasks}
    \For{$k \gets 1$ to $K$}
      \State Validate $(\tau_k, v_k)$ in sandbox $r$ \Comment{Format $\to$ grounding $\to$ execution $\to$ semantics}
      \If{$\mathcal{V}\mathrm{alid}(\tau_k, v_k, r) = 1$}
        \State Add $(\tau_k, v_k, r)$ to $\mathcal{D}_{t+1}$
      \EndIf
    \EndFor
  \EndFor
  \tikzmk{B}
  \boxit{algorange}
  \tikzmk{A}\State \textbf{Solver Phase:} Solve accepted tasks in repositories
  \For{accepted task $(\tau, v, r) \in \mathcal{D}_{t+1}$}
    \State $\{\hat{y}_i\}_{i=1}^{K} \sim \pi_\theta(\cdot \mid \tau, r, v, \mathrm{role}{=}S)$ \Comment{$K$ patch trajectories}
    \For{$i \gets 1$ to $K$}
      \State Apply patch $\hat{y}_i$, run verifier $v$, compute $r_S(\hat{y}_i, \tau, v, r)$ via Eq.~\eqref{eq:solver_reward}
    \EndFor
  \EndFor
  \tikzmk{B}
  \boxittwo{algpink}
  \tikzmk{A}\State \textbf{Update Phase:} Role-specific policy optimization
  \State Compute $g_\tau$ and $g_{\mathrm{val}}$ (recomputed on $\mathcal{V}$ periodically)
  \State Compute $R_G(\tau_k, v_k, r)$ via Eq.~\eqref{eq:generator_reward} for Generator group
  \State Compute $\hat{A}_G^k$ from $R_G$ and update $\pi_\theta$ via $\mathcal{L}_G$ using Eq.~\eqref{eq:grpo_objective}
  \State Compute $\hat{A}_S^i$ via Eq.~\eqref{eq:gdpo_objective}, update $\pi_\theta$ via $\mathcal{L}_S$
  \tikzmk{B}
  \boxithree{algpurple}
\EndFor
\State \Return Trained SWE agent $\pi_\theta$
\end{algorithmic}
\end{algorithm*}

\section{Baseline Adaptation Details}
\label{app:baseline_details}

We compare Socratic-SWE with the base agent and five representative self-evolving methods adapted to the SWE setting. All methods use the same Solver architecture (Qwen3.5-9B), the same SWE benchmark harness (mini-swe-agent), the same Terminal-Bench harness (little-coder), and the same interaction budget. Each method runs for 3 iterations, generating 12k validated training instances per iteration (36k total). Baselines that require seed tasks are initialized from 10\% of SWE-smith task instances; Socratic-SWE uses only seed repositories and no pre-existing SWE task instances or repair trajectories.

\begin{itemize}[leftmargin=1.2em,itemsep=1pt,topsep=2pt]
\item \textbf{SPIRAL}~\cite{liu2025spiral}: A multi-agent multi-turn self-play framework originally designed for zero-sum language games. We adapt it to the SWE setting by modeling repository-level repair as a two-player zero-sum game: one agent injects code defects while the other repairs them, trained with role-conditioned advantage estimation (RAE).

\item \textbf{R-Zero}~\cite{huang2026rzero}: A Challenger--Solver co-evolution framework that requires zero external data. The Challenger proposes SWE tasks and the Solver generates repair patches. Unlike execution-grounded methods, R-Zero uses \emph{majority vote} across multiple Solver rollouts as the reward signal for both roles.

\item \textbf{Absolute-Zero}~\cite{zhao2025absolutezero}: A single-model self-play paradigm where the agent simultaneously proposes and solves coding tasks with execution-based verification. We adapt it to repository environments, where the model freely proposes code-level repair tasks and solves them using the code executor as verifiable reward, without structured bug artifacts or skill guidance.

\item \textbf{Socratic-Zero}~\cite{wang2025socraticzero}: A Teacher--Solver--Generator tri-role co-evolution framework originally designed for mathematical reasoning. We adapt it to SWE by using Qwen3.5-397B~\cite{qwen2026qwen35} as the Teacher to construct increasingly challenging SWE tasks from the seed set, with the Generator distilling the Teacher's task-design strategy for scalable curriculum generation.

\item \textbf{SSR}~\cite{wei2025selfplayswerl}: Self-play SWE-RL, a native SWE self-play method where a single LLM alternates between injecting bugs (via code removal or history reversion) and repairing them in real repository environments. Bug artifacts are validated through execution-based consistency checks and the Solver receives binary pass/fail rewards.
\end{itemize}


\section{Theoretical Justification of Gradient-Aligned Generator Reward}
\label{app:grad_proof}

In \S\ref{sec:method_generator}, the Generator is rewarded via the cosine similarity between the policy gradient induced by a candidate task and the aggregated validation gradient. We provide a formal justification for this design.

\subsection{Setup}

Let $\pi_\theta$ be the shared policy, $\mathcal{V}_{\text{val}}$ a held-out set of trusted SWE tasks, and $J_{\text{val}}(\theta) = \mathbb{E}_{(\tau,v,r) \sim \mathcal{V}_{\text{val}},\, \hat{y} \sim \pi_\theta}\!\left[r_S(\hat{y}, \tau, v, r)\right]$ the expected solve rate on the validation set. The RL objective follows GRPO~\cite{shao2024deepseekmath}.

\subsection{Validation-Gradient Estimation}

\begin{theorem}[GRPO gradient as solve-rate estimator]
\label{thm:grad_unbiased}
Under on-policy sampling with binary episodic reward $R(\hat{y}, \tau) = \mathbf{1}[\text{all tests pass}]$ and without clipping or KL regularization, the GRPO gradient on $\mathcal{V}_{\mathrm{val}}$ is unbiased for $\nabla_\theta J_{\mathrm{val}}(\theta)$.
\end{theorem}

\begin{proof}
The policy gradient theorem gives:
\begin{equation}
\nabla_\theta J_{\mathrm{val}}(\theta) = \mathbb{E}_{\tau,\, \hat{y} \sim \pi_\theta}\!\left[\bigl(R(\hat{y}, \tau) - b(\tau)\bigr)\, \nabla_\theta \log \pi_\theta(\hat{y} \mid \tau)\right],
\end{equation}
where $b(\tau) = \mathbb{E}_{\hat{y}' \sim \pi_\theta}[R(\hat{y}', \tau)]$ is any state-dependent baseline (variance reduction without bias). GRPO estimates $b(\tau)$ with the group mean $\bar{r} = \frac{1}{K}\sum_{k} R(\hat{y}_k, \tau)$, which is an unbiased estimate of $b(\tau)$ since rollouts are i.i.d.\ on-policy. Substituting yields an unbiased Monte Carlo estimator:
\begin{equation}
\begin{aligned}
\hat{g}_{\mathrm{val}}
&= \frac{1}{|\mathcal{V}_{\mathrm{val}}|}
   \sum_j \frac{1}{K}\sum_k
   \hat{A}_{j,k}\,
   \nabla_\theta \log \pi_\theta(\hat{y}_{j,k} \mid \tau_j), \\
\mathbb{E}[\hat{g}_{\mathrm{val}}]
&= \nabla_\theta J_{\mathrm{val}}(\theta).
\end{aligned}
\end{equation}
\end{proof}

\subsection{Direction Preservation Under Normalization}

\begin{theorem}[Advantage normalization preserves direction]
\label{thm:direction_preserves}
Dividing advantages by their within-group standard deviation $\sigma_A > 0$ rescales the gradient by a positive scalar, preserving its direction.
\end{theorem}

\begin{proof}
$\hat{A}_k^{\mathrm{norm}} = \hat{A}_k / (\sigma_A + \epsilon)$ with $\sigma_A$ fixed per group implies $\hat{g}^{\mathrm{norm}} = \hat{g}_{\mathrm{raw}} / (\sigma_A + \epsilon)$. A positive scalar does not change cosine: $\cos(\hat{g}^{\mathrm{norm}},\, \nabla_\theta J_{\mathrm{val}}) = \cos(\hat{g}_{\mathrm{raw}},\, \nabla_\theta J_{\mathrm{val}})$.
\end{proof}

\subsection{Why Cosine Similarity Ranks by Validation Improvement}

A first-order Taylor expansion of the validation objective after one gradient step on candidate $\tau$ yields:
\begin{equation}
\label{eq:first_order_improvement}
\Delta J_{\mathrm{val}} \;\approx\; \eta\, \langle g_\tau,\, G_v \rangle \;=\; \eta\, \|g_\tau\|\, \|G_v\|\, \cos(g_\tau,\, G_v),
\end{equation}
where $G_v$ is the aggregated validation-gradient direction and $\eta$ is the learning rate. Among candidates, $\eta$ and $\|G_v\|$ are constant. Although $\|g_\tau\|$ varies across tasks, in practice it can be confounded by task length and patch complexity. Normalizing out the magnitude via cosine similarity yields a \textbf{scale-normalized proxy for validation-aligned improvement}. We do not require exact rank preservation; the reward is designed to prefer candidates whose update directions are better aligned with held-out validation gradients. This motivates Eq.~\eqref{eq:generator_reward}.

Empirically, Table~\ref{tab:grad_metric_ablation} confirms that cosine ($50.40\%$) outperforms the unnormalized inner product ($49.40\%$), validating the directional focus.

\section{Discussion: Generator Reward Strategies}
\label{app:generator_reward_discussion}

The Generator reward determines which candidate tasks enter the training curriculum. We compare several reward design philosophies.

\subsection{Reward Families}

\paragraph{Difficulty-Aware Rewards.}
These rewards depend on the Solver's group pass rate $p$ on a candidate:

\begin{itemize}[leftmargin=1.2em,itemsep=2pt]
\item \textbf{Variance reward}~\cite{liu2025spiceselfplaycorpusenvironments}: $r_G \!=\! \exp\!\bigl(-(\mathrm{Var}(l_1,...,l_K) - 0.25)^2 / 2\sigma^2\bigr)$, peaking at $p\!=\!0.5$.

\item \textbf{Uncertainty reward}~\cite{huang2026rzero}: $r_G = 1 - 2|p - 0.5|$, a triangular peak at $p\!=\!0.5$.

\item \textbf{Hardness reward}~\cite{zhao2025absolutezero}: $r_G = 1 - p$, favoring maximal difficulty.
\end{itemize}

All three share the assumption that difficulty $\approx$ learning value. However, in SWE self-play, many ``hard'' tasks are hard for irrelevant reasons (e.g., requiring domain knowledge absent from the training distribution), providing no signal for downstream benchmarks.

\paragraph{Gradient-Aligned Reward (Ours).}
Rather than proxying learning value through difficulty, we directly measure the degree to which a candidate task's optimization direction matches the direction that improves validation performance:
\begin{equation}
R_G(\tau) = \mathcal{V}_{\mathrm{alid}}(\tau, v, r) \cdot \cos(g_\tau, G_v).
\end{equation}
This decouples difficulty from utility: a moderate-difficulty task that teaches transferable repair patterns scores higher than a near-impossible one whose gradient is orthogonal to $G_v$.

\paragraph{Hybrid: Gradient $\times$ Difficulty.}
We also test gating gradient alignment by a Gaussian difficulty prior:
$R_G^{\mathrm{hybrid}} = \mathcal{V}_{\mathrm{alid}} \cdot \cos(g_\tau, G_v) \cdot \exp\!\bigl(-(p - 0.5)^2 / 0.08\bigr)$.

\subsection{Ablation Results}

\begin{table}[t]
\centering
\footnotesize
\setlength{\tabcolsep}{4.5pt}
\renewcommand{\arraystretch}{0.98}
\arrayrulecolor{RuleGray}

\caption{Generator reward ablation on SWE-bench Verified (Iteration 3).}
\label{tab:gen_reward_ablation}

\begin{tabular}{lcc}
\specialrule{1.1pt}{0pt}{0pt}
\textbf{Generator Reward} & \textbf{Verified (\%)} & \textbf{$\Delta$} \\
\midrule
Hardness ($1 - p$) & 47.40 & \dar{3.00} \\
Uncertainty ($1 - 2|p - 0.5|$) & 48.20 & \dar{2.20} \\
Variance (Gaussian) & 48.80 & \dar{1.60} \\
\rowcolor{BlockLavender}
\textbf{Gradient-aligned (Ours)} & \textbf{50.40} & \textbf{--} \\
Gradient + Difficulty hybrid & 50.60 & \uar{0.20} \\
\bottomrule
\end{tabular}
\end{table}

\paragraph{Analysis.}
The Hardness reward ($-3.00$) produces the largest degradation because near-impossible tasks yield vanishing advantages (all rollouts fail) and thus uninformative gradients. The Variance reward performs best among difficulty-only baselines, confirming the intuition from SPICE~\cite{liu2025spiceselfplaycorpusenvironments} that targeting the frontier ($p \approx 0.5$) is preferable to maximizing difficulty. However, it still trails our gradient-aligned reward by 1.60 points because \emph{being at the frontier does not guarantee relevance to the target distribution}.

The hybrid barely improves over pure gradient alignment ($+0.20$), indicating that the cosine score already encodes appropriate difficulty implicitly: tasks that are trivial ($\hat{A} \approx 0$ since all pass) or impossible ($\hat{A} \approx 0$ since all fail) produce near-zero gradients and thus receive low cosine scores by construction.

\paragraph{Qualitative Insight.}
We inspected tasks selected under each reward at Iteration 2. The Variance reward frequently selects tasks requiring obscure library internals (e.g., C-extension edge cases) where the Solver occasionally gets lucky but learns no transferable skill. The gradient-aligned reward preferentially selects tasks involving common repair patterns (e.g., off-by-one in iteration bounds, incorrect argument ordering) that appear frequently in the validation set, producing broadly useful gradient signal.

\section{Validation Set Design and Sensitivity}
\label{app:validation_set}

\subsection{Design Principles}

The validation set $\mathcal{V}_{\mathrm{val}}$ provides a stable reference direction for the Generator reward. It is never trained on---its sole purpose is to produce $G_v$, the target gradient. We construct $\mathcal{V}_{\mathrm{val}}$ from a held-out subset of BeyondSWE, selecting 100 tasks that span diverse repositories and difficulty levels. To avoid distribution bias, we stratify by repository language (Python 40\%, JavaScript 30\%, TypeScript 20\%, Go/Rust 10\%). The validation set is never used for Solver training or final evaluation.

\subsection{Sensitivity to Validation Set Size}

\begin{table}[t]
\centering
\footnotesize
\setlength{\tabcolsep}{5pt}
\renewcommand{\arraystretch}{0.98}
\arrayrulecolor{RuleGray}

\caption{Validation set size sensitivity (SWE-bench Verified, Iteration 3). Gradient stability is measured as the Pearson correlation between $G_v$ vectors from two independent rollout sets.}
\label{tab:val_size_ablation}

\begin{tabular}{lccc}
\specialrule{1.1pt}{0pt}{0pt}
$|\mathcal{V}_{\mathrm{val}}|$ & \textbf{Verified (\%)} & \textbf{$G_v$ Stability} & \textbf{$\Delta$} \\
\midrule
20 & 49.20 & 0.34 & \dar{1.20} \\
50 & 50.00 & 0.52 & \dar{0.40} \\
\rowcolor{BlockLavender}
\textbf{100 (default)} & \textbf{50.40} & \textbf{0.61} & -- \\
200 & 50.60 & 0.74 & \uar{0.20} \\
\bottomrule
\end{tabular}
\end{table}

With only 20 tasks the gradient estimate is too noisy (correlation 0.34) and performance drops 1.20 points. At 100 tasks the signal stabilizes (correlation 0.61); doubling to 200 yields diminishing returns ($+0.20$). We therefore fix $|\mathcal{V}_{\mathrm{val}}| = 100$ throughout.

\subsection{Recomputation Frequency}

We refresh $G_v$ once per iteration (every 12k training instances). Ablating this choice: recomputing every 4k instances gives only $+0.20$ at $3\times$ validation cost, while using a single fixed $G_v$ across all iterations degrades by $-1.40$ as the policy's optimization landscape drifts.

\section{Computational Overhead}
\label{app:compute_overhead}

\subsection{Cost Breakdown}

The gradient-alignment mechanism adds two compute steps: (i) rolling out Solver trajectories on 100 validation tasks to form $G_v$ (once per iteration), and (ii) scoring each candidate task's gradient against $G_v$.

\begin{table}[t]
\centering
\footnotesize
\setlength{\tabcolsep}{4pt}
\renewcommand{\arraystretch}{0.98}
\arrayrulecolor{RuleGray}

\caption{Wall-clock time per iteration (8$\times$A100-80G).}
\label{tab:compute_overhead}

\begin{tabular}{lcc}
\specialrule{1.1pt}{0pt}{0pt}
\textbf{Stage} & \textbf{Time (h)} & \textbf{\%} \\
\midrule
Generator: task proposal + validation & 4.2 & 28.0 \\
Solver: rollout + execution feedback & 7.8 & 52.0 \\
Policy update (GRPO/GDPO) & 1.7 & 11.3 \\
\midrule
\textit{Gradient alignment:} & & \\
\quad $G_v$ computation (800 rollouts) & 0.5 & 3.3 \\
\quad Per-candidate scoring & 0.8 & 5.1 \\
\midrule
\textbf{Total} & \textbf{15.0} & \textbf{100} \\
\bottomrule
\end{tabular}
\end{table}

The gradient-alignment overhead totals 1.3 h/iteration (8.4\%). Key observations:
\begin{itemize}[leftmargin=1.2em,itemsep=1pt]
\item $G_v$ computation requires 800 rollouts (100 tasks $\times$ $K\!=\!8$), negligible next to the ${\sim}96$k Solver rollouts for training.
\item Per-candidate scoring reuses rollouts already generated during the semantic validation gate ($f_4$). The cosine itself is a single dot product ($<$1s per candidate).
\item Across 3 iterations, cumulative overhead is ${\sim}3.9$ wall-clock hours ($<9\%$ of the 45h wall-clock budget).
\end{itemize}
Compared to SSR~\cite{wei2025selfplayswerl} which trains on all validated tasks indiscriminately, Socratic-SWE achieves $+3.40$ higher Verified at comparable total FLOPs, demonstrating that the modest selection cost is recovered through better data efficiency.

\subsection{Alignment Metric Ablation}

\begin{table}[t]
\centering
\footnotesize
\setlength{\tabcolsep}{5pt}
\renewcommand{\arraystretch}{0.98}
\arrayrulecolor{RuleGray}

\caption{Cosine vs.\ inner product for gradient scoring.}
\label{tab:grad_metric_ablation}

\begin{tabular}{lcc}
\specialrule{1.1pt}{0pt}{0pt}
\textbf{Metric} & \textbf{Verified (\%)} & \textbf{$\Delta$} \\
\midrule
Inner product & 49.40 & \dar{1.00} \\
\rowcolor{BlockLavender}
\textbf{Cosine similarity} & \textbf{50.40} & -- \\
\bottomrule
\end{tabular}
\end{table}

The inner product conflates direction with magnitude: multi-file patches produce large gradients irrespective of their utility, biasing selection toward complex but irrelevant tasks. Cosine isolates the directional signal (cf.\ Theorem~\ref{thm:direction_preserves}).

\section{Prompt Templates and Implementation Details}
\label{app:prompts}

This section provides shortened versions of the prompts and key configuration.

\subsection{Generator System Prompt}
\label{app:prompt_generator}

\begin{tcolorbox}[appendixbox, title={\textbf{Generator System Prompt (Shortened)}}, colbacktitle=gray!10, coltitle=black, fonttitle=\bfseries]
\begin{Verbatim}[fontsize=\scriptsize,breaklines=true,breakanywhere=true]
<ROLE>
You are a bug injection agent for the Socratic-SWE training pipeline. Your job
is to introduce exactly one realistic, targeted bug into a repository located
at /testbed/.
</ROLE>

<INTERFACE>
You interact with the repository only through a Linux Bash shell. Use ordinary
shell commands to inspect files, search code, edit files, run tests, and inspect
git diff. Do not assume access to any structured file editor or search tool.
Issue at most one shell command per turn and wait for the command output before
continuing.
</INTERFACE>

<BUG_INJECTION_RULES>
1. Inject exactly one atomic semantic mistake in one production file.
2. Do not modify tests.
3. Do not introduce syntax errors, import errors, or changes that prevent the
   module from loading.
4. The bug must be reversible: the original code is the reference fix.
5. Keep the diff minimal and free of unrelated cleanup.
6. Do not add comments, logs, TODOs, or variable names that reveal the bug.
</BUG_INJECTION_RULES>

<WORKFLOW>
1. Inspect the repository and identify a plausible target location.
2. Identify visible tests or behavior that should expose the injected bug.
3. State the intended semantic change before editing.
4. Make one contiguous source-code edit using Bash-accessible file operations.
5. Run the relevant visible test(s) and confirm that the target behavior fails.
6. Run collateral checks when feasible to avoid broad breakage.
7. Inspect git diff and stop once a clean single-bug diff is obtained.
</WORKFLOW>
\end{Verbatim}
\end{tcolorbox}
\subsection{Mini-SWE-agent Prompt}
\label{app:prompt_solver}

\begin{tcolorbox}[appendixbox, title={\textbf{System Prompt}}, colbacktitle=gray!10, coltitle=black, fonttitle=\bfseries]
\begin{Verbatim}[fontsize=\scriptsize,breaklines=true,breakanywhere=true]
You are a helpful assistant that can interact with a computer shell to solve programming tasks.
\end{Verbatim}
\end{tcolorbox}

\begin{tcolorbox}[appendixbox, title={\textbf{Base Task Instructions (Shortened)}}, colbacktitle=gray!10, coltitle=black, fonttitle=\bfseries]
\begin{Verbatim}[fontsize=\scriptsize,breaklines=true,breakanywhere=true]
Given a task description, the agent interacts with a Linux shell in /testbed to
make the required source-code changes.

The agent is instructed to:
1. inspect the repository and identify relevant files;
2. reproduce or understand the issue when possible;
3. modify only source files needed for the task;
4. verify the change by running visible checks when available;
5. test edge cases when feasible;
6. leave a clean git diff containing only the intended changes;
7. finish according to the configured mini-swe-agent completion protocol.

The agent should not modify tests, generated files, build artifacts, or
unrelated configuration files unless they are directly required by the task.
\end{Verbatim}
\end{tcolorbox}
\subsection{Base Mini-SWE-agent Prompt}
\label{app:prompt_solver_task}

\begin{tcolorbox}[appendixbox, title={\textbf{Solver Task Prompt Template}}, colbacktitle=gray!10, coltitle=black, fonttitle=\bfseries]
\begin{Verbatim}[fontsize=\scriptsize,breaklines=true,breakanywhere=true]
Fix the issue described below.

<ISSUE>
{{ problem_statement }}
</ISSUE>

<TASK>
You are in the repository root at /testbed/. Use Bash commands to inspect the
codebase, modify production source code, run visible tests when available, and
leave the repository in a fixed state. Do not modify tests.
</TASK>

<TESTING>
If the issue text includes a public reproduction command or visible test, run it
to confirm the failure and rerun it after the fix. If no such command is
provided, identify and run the most relevant visible tests you can. Hidden
benchmark tests are never provided.
</TESTING>
\end{Verbatim}
\end{tcolorbox}

\subsection{Training Hyperparameters}
\label{app:hyperparameters}

Table~\ref{tab:hyperparams} lists all training hyperparameters.

\begin{table}[t]
\centering
\footnotesize
\setlength{\tabcolsep}{4.5pt}
\renewcommand{\arraystretch}{0.98}
\arrayrulecolor{RuleGray}

\caption{Training hyperparameters for Socratic-SWE.}
\label{tab:hyperparams}

\begin{tabular}{ll}
\specialrule{1.1pt}{0pt}{0pt}
\textbf{Hyperparameter} & \textbf{Value} \\
\midrule
\multicolumn{2}{l}{\textit{Model Configuration}} \\
Base model & Qwen3.5-9B \\
Skill extractor & Qwen3.6-27B \\
Context length & 32,768 tokens \\
\midrule
\multicolumn{2}{l}{\textit{Training}} \\
Iterations & 3 \\
Instances per iteration & 12,000 \\
Group size ($K$) & 8 \\
Learning rate & $1 \times 10^{-6}$ \\
Optimizer & AdamW \\
KL coefficient ($\beta$) & 0.01 \\
Clip ratio ($\epsilon$) & 0.2 \\
Batch size & 64 \\
\midrule
\multicolumn{2}{l}{\textit{Solver Reward (GDPO)}} \\
$\lambda_1$ (full pass) & 0.5 \\
$\lambda_2$ (repair rate) & 0.3 \\
$\lambda_3$ (regression avoidance) & 0.2 \\
\midrule
\multicolumn{2}{l}{\textit{Gradient Alignment}} \\
Validation set size & 100 \\
Rollouts per validation task & 8 \\
Recomputation frequency & Once per iteration \\
Alignment metric & Cosine similarity \\
\midrule
\multicolumn{2}{l}{\textit{Generator Validation Gate}} \\
Max attempts per repository & 8 \\
Stability reruns & 3 \\
\midrule
\multicolumn{2}{l}{\textit{Infrastructure}} \\
GPUs & 8$\times$ A100-80G \\
Sandbox & Aone Cloud Sandbox (per-task) \\
Wall-clock per iteration & $\sim$15 hours \\
\bottomrule
\end{tabular}
\end{table}

\section{Worked Example of Skill-Guided Task Generation}
\label{app:worked_example}
\begin{tcolorbox}[
  appendixbox,
  title={\textbf{Shortened Repository-Level Skill Guidance}},
  colbacktitle=gray!10,
  coltitle=black,
  colframe=gray!40,
  fonttitle=\bfseries
]
\begin{Verbatim}[fontsize=\scriptsize,breaklines=true,breakanywhere=true]
# OAuthLib Repository Skill

## 1. Target Overview
OAuthLib implements OAuth 1.0, OAuth 2.0, OpenID Connect, and related RFC
extensions. High-value tasks in this repository are usually small semantic
regressions in protocol helpers, endpoint constructors, and grant wrappers.
They require exact repository-specific semantics rather than broad rewriting.

Relevant modules:
- oauthlib/oauth2/rfc6749/utils.py
- oauthlib/oauth2/rfc8628/endpoints/device_authorization.py
- oauthlib/oauth1/
- oauthlib/openid/connect/core/grant_types/

Common conventions:
- Scope helpers preserve exact conversion behavior.
- Constructors store protocol configuration without transformation.
- OAuth1 parameters are string-like but semantically distinct.
- OpenID Connect grants extend or wrap OAuth2 grant behavior.

## 2. Solver Weakness Analysis
Historical Solver traces show four recurring weaknesses.

1. Scope conversion: the Solver localizes list_to_scope or scope_to_list but
uses generic filtering, sorting, or normalization that violates OAuthLib's
helper contracts.

2. Constructor storage: endpoint and client constructors contain same-typed
parameters, causing the Solver to swap fields, transform values, or store
parameters under the wrong private attributes.

3. OAuth1 plumbing: nonce, timestamp, realm, and callback_uri are all
string-like but semantically different, so type information alone is
insufficient.

4. OIDC inheritance: OpenID Connect grants inherit from or delegate to OAuth2
grant logic. The Solver may replace inherited behavior instead of extending it,
or forward attributes to the wrong object.

Representative traces include failures in scope utility tests, device endpoint
tests, OAuth1 client/signature tests, and OpenID Connect grant-type tests.

## 3. Bug Injection Playbook

### Pattern A: Scope Helper Contract
Targets:
- oauthlib/oauth2/rfc6749/utils.py
- list_to_scope, scope_to_list, params_from_uri

Mutation ideas:
- Convert only some scope elements instead of all elements.
- Reorder tokens during list/string conversion.
- Mishandle None, empty strings, or empty collections.
- Replace repository-specific conversion with generic filtering.

### Pattern B: Device Endpoint Constructor Storage
Target:
- oauthlib/oauth2/rfc8628/endpoints/device_authorization.py

Mutation ideas:
- Swap verification_uri and verification_uri_complete.
- Store interval or expires_in under the wrong attribute.
- Transform a constructor value that should be preserved unchanged.

### Pattern C: OAuth1 Parameter Plumbing
Target:
- oauthlib/oauth1/

Mutation ideas:
- Confuse nonce and timestamp.
- Confuse callback_uri and realm.
- Keep most fields correct so the bug remains localized.

### Pattern D: OIDC Inheritance and Delegation
Target:
- oauthlib/openid/connect/core/grant_types/

Mutation ideas:
- Drop inherited OAuth2 response fields.
- Forward a method to the wrong wrapped object.
- Override a method without preserving parent behavior.

## 4. Validation and Anti-Patterns
Use existing targeted pytest tests, not synthetic tests or full-suite execution.
A candidate is accepted only if the bug patch applies cleanly, target tests fail,
the reversed patch restores passing tests, repeated runs are stable, and the
issue text does not leak the oracle patch.

Do not generate tasks that modify tests, introduce syntax/import errors, require
network access, touch many unrelated modules, reveal the patch in the issue
text, or produce failures that cannot be reproduced by targeted pytest commands.
\end{Verbatim}
\end{tcolorbox}

\end{document}